\documentclass[12pt]{article}
\usepackage{amsmath}
\usepackage{graphicx}
\usepackage{enumerate}
\usepackage{natbib}
\usepackage{url} 
\usepackage{ulem}
\usepackage{algorithm}
\usepackage{algpseudocode}
\usepackage{soul}
\usepackage{longtable}
\usepackage{booktabs}

\usepackage{comment}

\newcommand{\ind}{\perp\!\!\!\!\perp} 
\usepackage[utf8]{inputenc}
\usepackage{longtable}
\usepackage{booktabs}
\usepackage{multirow}
\usepackage{bm}
\usepackage{amsmath, amsfonts, amsthm, amssymb, amsbsy, graphicx, dsfont, mathtools, enumerate, enumitem}
\usepackage{graphicx}
\usepackage[english]{babel}
\usepackage{graphicx}
\usepackage{subfigure, caption}
\usepackage{float}
\usepackage{geometry}
\usepackage{array}
\setcitestyle{authoryear,open={(},close={)}}
\usepackage[CJKbookmarks=true,
            bookmarksnumbered=true,
			bookmarksopen=true,
			colorlinks=true,
			citecolor=blue, 
			linkcolor=blue,
			anchorcolor=red,
			urlcolor=blue]{hyperref}
\usepackage[usenames]{color}

\usepackage{prettyref,soul}
\usepackage[font=small,labelfont=it]{caption}
\RequirePackage[framemethod=TikZ]{mdframed}

\newtheorem{theorem}{Theorem}[section]

\newtheorem{proposition}[theorem]{Proposition}

\theoremstyle{definition}

\theoremstyle{remark}
\newtheorem*{remark}{Remark}

\newcommand{\argmin}{\mathop{\rm argmin}}

\def\X{\mathbf{X}}

\def\Prob{\mathbb{P}}

\def\bSigma{\boldsymbol{\Sigma}}

\def\wt{\widetilde}
\def\wh{\widehat}

\def\mA{\mathbb{A}}

\def\mE{\mathbb{E}}

\def\mR{\mathbb{R}}

\def\mX{\mathbb{X}}

\def\mY{\mathbb{Y}}

\def\cN{\mathcal{N}}

\def\cS{\mathcal{S}}

\def\ind{\perp\!\!\!\perp}
\def\mX{\mathbf{X}}
\def\rX{X}

\def\mY{\mathbf{Y}}
\def\rY{Y}

\def\rW{W}

\def\mA{\mathbf{A}}
\def\rA{A}

\def\mAt{\wt{\mathbf{A}}}
\def\rAt{\wt{A}}

\def\mXAt{\wt{\mathbf{X}}^A}
\def\rXAt{\wt{X}^A}
\def\vXAt{\wt{x}^A}

\def\eqdist{\overset{\mathcal{D}}{=}}

\begin{document}

\def\spacingset#1{\renewcommand{\baselinestretch}%
{#1}\small\normalsize} \spacingset{1.25}


\title{Gold after Randomized Sand: \\ Model-X Split Knockoffs for Controlled Transformation Selection 
}
\author{Yang Cao$^\star$, Hangyu Lin$^\dagger$, Xinwei Sun$^\ddagger$, Yuan Yao$^\dagger$\\
    ~\\
    $^\star$\textit{Yale University} \\
    $^\dagger$\textit{The Hong Kong University of Science and Technology} \\
    $^\ddagger$\textit{Fudan University}
}

\date{}
\maketitle

\begin{abstract}
    Controlling the False Discovery Rate (FDR) is critical for reproducible variable selection, especially given the prevalence of complex predictive modeling. The recent Split Knockoff method, an extension of the canonical Knockoffs framework, offers finite-sample FDR control for selecting sparse transformations but is limited to linear models with fixed designs. Extending this framework to random designs, which would accommodate a much broader range of models, is challenged by the fundamental difficulty of reconciling a random covariate design with a deterministic linear transformation. To bridge this gap, we introduce Model-X Split Knockoffs. Our method achieves robust FDR control for transformation selection in random designs by introducing a novel auxiliary randomized design. This key innovation effectively mediates the interaction between the random design and the deterministic transformation, enabling the construction of valid knockoffs. Like the classical Model-X framework, our approach provides provable, finite-sample FDR control under known or accurately estimated covariate distributions, regardless of the response's conditional distribution. Importantly, it guarantees at least the same, and often superior, selection power as standard Model-X Knockoffs when both are applicable. Empirical studies, including simulations and real-world applications to Alzheimer's disease imaging and university ranking analysis, demonstrate robust FDR control and improved statistical power.
\end{abstract}

\section{Introduction}

Variable selection is a fundamental problem in statistical research, focusing on identifying which variables from a large candidate pool significantly influence an outcome. A key challenge in this area is ensuring reproducibility and replicability, as the inclusion of more predictors often increases predictive power but risks over-selection in finite samples. To address this, controlling the False Discovery Rate (FDR) has become a central concern. The seminal work of \cite{benjamini1995controlling} introduced the concept of FDR and proposed a method to control it in settings with independent hypothesis tests or orthogonal designs. This groundbreaking contribution laid the foundation for extensive theoretical and practical advancements in FDR control \citep{benjamini2001control, storey2002direct, reiner2003identifying, benjamini2010discovering}.

Building on this, \cite{barber2015controlling} introduced the Knockoff method, which extended FDR control with finite-sample guarantees to general fixed design settings, including sparse linear models with arbitrary design matrices. This represented a significant advance over earlier methods, which were often limited to specific design classes. The Knockoff framework has since inspired considerable research \citep{dai2016knockoff, xu2016false, barber2019knockoff}. 

Further generalization was recently achieved by \cite{cao2024controlling, cao2024split}, who introduced the concept of \emph{transformational sparsity}. While traditional Knockoff methods focus on sparsity in the original measured covariates, transformational sparsity considers sparsity in \emph{predetermined} transformations (often linear) of the parameters. This paradigm is particularly relevant in applications where sparsity manifests in derived features rather than raw variables. For example, in genomic research, disease mechanisms often involve biological pathways (linear combinations of gene effects) rather than isolated genes \citep{ramanan2012pathway}; in neuroimaging, wavelet transforms are applied to MRI data to extract meaningful patterns \citep{donoho1995adapting, bullmore2004wavelets}; and in economics, temporal contrasts like trend filtering are often analyzed instead of raw features \citep{kim2009ell_1}. To address this, \cite{cao2024controlling} introduced the Split Knockoff framework, which extends the original Knockoff method to control the false discovery rate (FDR) with finite-sample guarantees under transformational sparsity in fixed design settings. This framework was further generalized to control directional FDR in \cite{cao2024split}, which finds successful applications in Alzheimer's disease image analysis and statistical ranking with multiple comparisons.

In another direction, \cite{candes2018panning} introduced the Model-X Knockoff framework, which enables FDR control in random design settings. This method leverages knowledge of the marginal distribution of the design matrix and identifies the \emph{Markov blanket} -- the smallest subset of variables that renders the outcome independent of all others -- allowing variable selection in a wide range of statistical models. The Model-X framework represents a significant advance, accommodating diverse modeling scenarios and inspiring substantial research \citep{barber2020robust, romano2020deep, ren2023knockoffs, ren2024derandomised}.

However, the Model-X framework is not directly compatible with transformational sparsity, where sparsity is defined with respect to deterministic transformations of the parameters rather than the associated variables themselves. For instance, in pairwise comparison problems, identifying that an outcome depends on two variables (as part of the Markov blanket) does not indicate whether these variables exert the same or different effects on the outcome. While Model-X Knockoffs identify influential variables, transformational sparsity concerns specific relationships or null hypotheses about transformed parameters, such as testing the equality of effects.

To address this gap, this paper introduces the \emph{Model-X Split Knockoff} method, which extends the Model-X Knockoff framework to handle transformational sparsity in random design settings. This novel method provides FDR control with finite-sample guarantees under transformational sparsity across a broad range of statistical models, beyond the linear settings considered in earlier works \citep{cao2024controlling, cao2024split}. Moreover, as will be demonstrated, the Model-X Split Knockoff method achieves at least the same selection power as the original Model-X Knockoff method and often outperforms it in empirical applications when both methods are applicable.

Specifically, this paper considers models where the conditional expectation of the response $Y\in \mathbb{R}$ given the \emph{random} designs $\rX:= (\rX_1, \dots, \rX_p)\in \mathbb{R}^p$ subject to the marginal distribution $P^*_{\rX}$, as well as the \emph{deterministic} transformation of parameters, is defined as follows:
\begin{align}
\label{eq: gen_int}
\mathbb{E}(Y|X) = f(X\beta^*), \quad \gamma^* = D\beta^*,
\end{align}
where $f: \mathbb{R}\to \mathbb{R}$ is the link function, and $D \in \mathbb{R}^{m \times p}$  is a predetermined linear transformation matrix that induces sparsity in the transformed coefficient vector
$\gamma^* \in \mathbb{R}^m$. Transformation selection focuses on identifying the support set of 
$\gamma^*$, i.e. $\mathcal{H}_{1}^{\gamma} = \{i : \gamma^*_i \neq 0\}$, rather than directly identifying the support of $\beta^*\in \mathbb{R}^p$. For example, in multiple comparison problems, $D$ is the graph difference matrix such that $(D\beta^*)_{(j,k)}:=\beta^*_j-\beta^*_k$ for each pair $(j,k)$ in an edge set. In this case, $\gamma^*_{(j, k)}=0$ corresponds to the null hypothesis $\mathcal{H}^{j,k}_0: \beta_j^*=\beta_k^*$, meaning that covariate $X_j$ has the same conditional effect as $X_k$. 

A significant challenge in working with model \eqref{eq: gen_int} arises from its inherent structure, which combines the \emph{random} design $X$ with a \emph{deterministic} transformation $D$, ultimately defining $\gamma^*$, whose sparsity is the primary focus. In this transformational sparsity scenario, there is no straightforward design matrix for $\gamma^*$ to directly work with, as neither $X$ alone nor $D$ directly maps $\gamma^*$ to the observed outcome $Y$. Instead, the randomness of $X$ becomes entangled with the deterministic transformation $D$, resulting in an implicit subspace constraint. This unique scenario poses challenges that lie outside the scope of the traditional Model-X framework, introducing substantial complexities for both methodological development and theoretical analysis.

To address these challenges, we adopt the Split Knockoffs approach introduced by \cite{cao2024controlling}, which considers the lifted parameter space $(\beta,\gamma)$ by relaxing the subspace constraint in \eqref{eq: gen_int} to its Euclidean neighborhood. To reconcile the random design $X$ with the deterministic design implicitly imposed by the subspace constraints, we introduce an auxiliary random vector $\rA \in \mathbb{R}^m$. This random vector is specifically constructed to map the constraint on $\gamma^*$ and its neighborhood to the measurement associated with the response, effectively introducing a randomized design associated with $\gamma^*$ into the model. This facilitates the construction of Model-X Split Knockoffs for inference on the transformed parameters $\gamma^*$. Specifically, the auxiliary randomized design $\rA$ transforms $\rX\beta^*$ into the following random design on the lifted parameter $(\beta^*,\gamma^*)$:
\begin{align}
\label{eq: decomposition}
    \rX\beta^*=\rXAt\beta^* + \rA\gamma^*, \quad \rXAt:=X-AD.
\end{align}
This decomposition holds because $\rA\gamma^* = \rA D\beta^*$ for all $\rA$. In this formulation, $\beta^*$ associated with the transformed random design $\rXAt = \rX - \rA D$, while $\gamma^*$ --- which captures the transformational sparsity --- is linked to the auxiliary randomized design $\rA$. This decomposition enables us to design Model-X split knockoffs for $\rA$ conditional on $\rXAt$, provided some exchangeability and independence conditions are satisfied (see Section \ref{sec: construction} for details). Provable false discovery rate (FDR) control is guaranteed as long as $P^*_X$ is known, bootstrapped, or can be accurately estimated (see Section \ref{sec.fdrcontrol}). Furthermore, the randomized design $\rA$ and the relaxation of the subspace constraint, from $\gamma^* = D\beta^*$ to its Euclidean neighborhood, can be leveraged to enhance selection power (see Section \ref{sec: w-statistics} and discussions therein).  

Since the proposed method is grounded in the Model-X framework and does not require specifying the conditional distribution  $P_{\rY|\rX}$, it accommodates more complex models beyond simple linear ones.  Additionally, this separation provides greater flexibility to optimize selection power compared to the original Model-X Knockoff method when $D = I$, as it imposes no assumptions on the distribution of $\rA$. Benefiting from this flexibility, the Model-X Split Knockoff method is guaranteed to achieve at least the same power as the Model-X Knockoff method, and empirical results demonstrate that it often achieves superior selection power in practice.

While the choice of the auxiliary random vector $\rA$ could, in principle, be arbitrary, subsequent steps in the Model-X Split Knockoff procedure are substantially dependent on the joint distribution of $(X, \rA)$. Although a general algorithm that can operate with an arbitrary joint distribution of $(X, \rA)$ exists, this paper will introduce specific constructions of $\rA$ tailored to canonical examples, such as pairwise comparisons. The motivation for developing these tailored constructions is to facilitate more efficient computation and to improve statistical power, particularly in scenarios where the distribution of $X$ is unknown and must be estimated from data.

To evaluate the performance of our proposed method, we conduct experiments on both simulated and real-world datasets, including the Alzheimer’s Disease and world college rankings. Simulation experiments using both regression and pairwise comparison data from random graphs validate the method's effectiveness. For the Alzheimer’s Disease dataset, the method successfully identifies abnormal regions and connections associated with the disease. Furthermore, it identifies important and significant gaps between college pairs in the world college ranking dataset.

We introduce the Model-X Split Knockoff methodology in Section~\ref{sec.method} and its theoretical results for FDR control in Section~\ref{sec.fdrcontrol}. Simulation experiments are detailed in Section~\ref{sec.simu}, followed by applications on real-world datasets in Section~\ref{sec.alz} and Section~\ref{sec.college}.

\section{Methodology}
\label{sec.method}

In this paper, we address the transformational sparsity problem, modeled in Equation \eqref{eq: gen_int} and quoted here for convenience: $\mE(\rY|\rX) = f(\rX\beta^*)$ and $\gamma^* = D\beta^*$.  
The primary goal of this paper is to estimate the support set of $\gamma^*$, defined as $\mathcal{H}^{\gamma}_{1} = \{i : \gamma^*_i \ne 0\}$, while controlling the false discovery rate. A key challenge arises from the interplay of a random design $\rX$ and a deterministic transformation $D$, where neither directly maps $\gamma^*$ to the outcome $\rY$.

To overcome this, we draw inspiration from the variable splitting technique recently proposed by \cite{cao2024controlling}. We introduce a new auxiliary design for $\gamma^*$ and integrate it into the model to facilitate inference on $\gamma^*$. The specifics of constructing this valid design are detailed below.

\subsection{Variable Splitting and Auxiliary Randomized Design}
An auxiliary random vector $\rA = (\rA_1, \dots, \rA_{m})$  is generated as in \eqref{eq: decomposition} to manually construct the design for $\gamma^*$, enabling the reformulation of model \eqref{eq: gen_int} as:
\begin{align}
\label{eq:decomp}
    \mE(Y|\rXAt, A)  = f(\rXAt\beta^* + \rA \gamma^*).
\end{align}
Recall that the equation holds because $\rXAt: =\rX-\rA D$ and $A\gamma^* = AD\beta^*$.

In this way, we successfully inserted an auxiliary $\rA$ for the parameter of interest $\gamma^*$, and lifted the model \eqref{eq: gen_int} into a higher-dimensional space including both $\beta^*$ and $\gamma^*$. With Equation \eqref{eq:decomp}, a split knockoff copy $\rAt$ could be constructed following certain necessary constraints to allow variable selection with the control of the false discovery rate in $\gamma^*$. After such a construction, standard Split Knockoff procedures as described in \cite{cao2024controlling} could be applied to finalize the procedure. In the following section, we will proceed to the step of constructing the Split Knockoff copy.

\subsection{Construction of Model-X Split Knockoffs}

\label{sec: construction}
Similarly to \citep{candes2018panning, cao2024controlling}, the Model-X Split Knockoff requires the construction of Split Knockoff copies to achieve the false discovery rate control to be shown in Section \ref{sec.fdrcontrol}. The random\footnote{To account for the randomness introduced by our construction, one can apply the e-value scheme of \citep{ebhq, lee2024boosting} detailed in Supplementary Material Section H.} vector $\rAt = (\rAt_1, \dots, \rAt_m)$ is called a Split Knockoff copy with respect to $(\rXAt, \rA)$ if the following conditions hold.
\begin{itemize}
    \item Exchangeability: for any subset $S \subseteq  \{1, \dots, p\}$,
    \begin{align}
         (\rXAt, \rA, \rAt)_{\text{swap}(S)} \eqdist (\rXAt, \rA, \rAt)\label{eq: exchangeability}.
    \end{align}
    \item Conditional independence: $\rAt \ind \rY | (\rXAt, \rA)$.
\end{itemize}
Here, $ (\rXAt, \rA, \rAt)_{\text{swap}(S)}$ means to swap the features in $S$ between $\rA$ and $\rAt$. For example, if $\rA \in \mR^{2}, S = \{1\}$, then $(\rXAt, \rA_1, \rA_2,\rAt_1, \rAt_2)_{\text{swap}(S)} = (\rXAt, \rAt_1, \rA_2, \rA_1, \rAt_2)$.

In the cases where the distribution of $\rX$ is well known, for any choice of $\rA$, $\rAt$ can be constructed sequentially to satisfy the exchangeability condition \eqref{eq: exchangeability}, as shown in Algorithm \ref{algo: sequential generation}, which is similar to \cite{candes2018panning}. The proof of Algorithm \ref{algo: sequential generation} satisfies the exchangeability is given in Supplementary Material Section A. 

\begin{algorithm}
\caption{Sequential Generation of Split Knockoff Copies}
\label{algo: sequential generation}
\begin{algorithmic}
\Require the distribution of $\rX$
\State \textbf{Generate} random variable $\rA$ and the respective random variable $\rXAt$
\For{$j = 1: m$}
    \State sample $\rAt_j$ from the conditional distribution: $\mathcal{L}(\rA_j|\rA_{-j}, \rAt_{1:j-1}, \rXAt)$
\EndFor
\end{algorithmic}
\end{algorithm}

Although Algorithm \ref{algo: sequential generation} gives the general construction of the Split Knockoff copy $\rAt$, it might face computational challenges in practice as the computational cost increases rapidly as $m$ increases. Moreover, the generated Split Knockoff copies can be highly correlated, which harms the selection power. In the following, we deal with two canonical examples to present refined constructions of Split Knockoff copies, which leads to lower computational costs and better selection power.

\subsubsection{Construction for Normal Distributions}

\label{sec: normal}

In the cases where the distribution of $\rX$ follows or can be well approximated by the normal distribution $\mathcal{N}(0, \Sigma_{\rX})$, the natural choice of the pair $(\rA, \rAt)$ is to generate it from a joint normal distribution of $(\rX, \rA, \rAt)$, as suggested by \cite{candes2018panning}. Especially, as we will show later, among all the choices of joint normal distributions, the choices that make the ``intercept'' $\rXAt$ independent from the auxiliary design $\rA$ are expected to achieve higher selection power.

The generation process begins with $\rA$ drawn from 
$(\rX, \rA) \sim \mathcal{N}(\mu, \Sigma)$, 
where $\mu = \mathbf{0}_{m+p}$ and 
$\Sigma = (\Sigma_{\rX} \;\; \Sigma_{\rA\rX}^\top;\;\; \Sigma_{\rA\rX} \;\; \Sigma_{\rA})$, with $\Sigma_{\rA}$ and $\Sigma_{\rA\rX}$ being two suitable matrices which could be used to optimize the selection power. With some calculations detailed in Supplementary Material Section F,  
conditional on the observation $\rXAt = \vXAt$, the Split Knockoff copy $\rAt$ could be generated from the joint distribution
\begin{align}
\label{eq: A-tilde-general}
    (\rA, \rAt)|(\rXAt = \vXAt) \sim \mathcal{N}\left\{ \left( \begin{array}{c} \wt{\mu}_{\vXAt} \\ \wt{\mu}_{\vXAt} \end{array} \right), \left( \begin{array}{cc} \wt{\Sigma}_{\rA} & \wt{\Sigma}_{\rA} - \mathrm{diag}(\mathbf{s}) \\ \wt{\Sigma}_{\rA} - \mathrm{diag}(\mathbf{s}) & \wt{\Sigma}_{\rA} \end{array} \right) \right\},
\end{align}
where $\mathbf{s}\in \mathbb{R}^m$ is a non-negative vector which makes the above covariance matrix semi-positive definite, $\wt{\mu}_{\vXAt} := (\Sigma_{\rA\rX} -\Sigma_{\rA} D)\Sigma_{\rXAt}^{-1}(\vXAt)^\top$ with $\Sigma_{\rXAt} := \Sigma_{\rX} - \Sigma_{\rX\rA}^\top D-\Sigma_{\rX\rA}D^\top +D^\top \Sigma_{\rA}D$, and $\wt{\Sigma}_{\rA} := \Sigma_{\rA} - (\Sigma_{\rA\rX} -\Sigma_{\rA} D)\Sigma_{\rXAt}^{-1}(\Sigma_{\rA\rX}^\top -D^\top \Sigma_{\rA})$. One can verify that the exchangeability \eqref{eq: exchangeability} holds from the joint distribution of $(\rX, \rA, \rAt)$.

We are going to specify $\Sigma_{\rA}$ and $\Sigma_{\rA\rX}$ in order to optimize the selection power.
\begin{enumerate}
    \item[1]\textbf{Choice of $\Sigma_{\rA\rX}$.} As discussed in \citep{barber2015controlling, candes2018panning}, the vector $\mathbf{s}$  in Equation \eqref{eq: A-tilde-general} should be as large positive as possible to maximize the selection power. It could be verified that $\mathbf{s}$ needs to satisfy 
    $\mathrm{diag}(\mathbf{s})\preceq 2 \wt{\Sigma}_{\rA}\preceq 2\Sigma_{\rA},$
    where the equality of the second inequality occurs if and only if $\Sigma_{\rA\rX} = \Sigma_{\rA} D$. This gives us evidence that taking $\Sigma_{\rA\rX} = \Sigma_{\rA} D$, i.e. $\rXAt\perp\rA$ as $cov(\rXAt,\rA) = \Sigma_{\rA\rX} - \Sigma_{\rA} D = 0$, helps improve the selection power.
    \item[2]\textbf{Choice of $\Sigma_{\rA}$.}  Another intuition on improving the selection power is to improve the incoherence conditions as mentioned in \cite{cao2024controlling}. In this spirit, one shall take $\Sigma_{\rA} = \alpha I_m$ for some $\alpha>0$ to improve the selection power.
\end{enumerate}

In summary, the choice of $\Sigma_{\rA}$ and $\Sigma_{\rA\rX}$ to optimize the selection power is $\Sigma_{\rA} = \alpha I_m$ and $\Sigma_{\rA\rX} = \alpha D$, where $\alpha>0$ satisfies $\Sigma_{\rX} - \alpha D^\top  D \succeq 0$. In this case, let the non-negative vector $\mathbf{s}$ satisfy $\mathrm{diag}(\mathbf{s})\preceq2\alpha I_m$, the joint distribution of $(\rX, \rA, \rAt)$ is
\begin{align}
\label{eq.A-tilde-general}
    \mathcal{N}\left\{ \left( \begin{array}{c} \mathbf{0}_p \\ \mathbf{0}_m\\ \mathbf{0}_m \end{array} \right), \left( \begin{array}{ccc} \Sigma_{\rX} & \alpha D^\top  &  \alpha D^\top  - \mathrm{diag}(\mathbf{s})D^\top   \\ \alpha D & \alpha I_m & \alpha I_m - \mathrm{diag}(\mathbf{s})\\  \alpha D - \mathrm{diag}(\mathbf{s})D & \alpha I_m - \mathrm{diag}(\mathbf{s}) & \alpha I_m\end{array}  \right) \right\}.
\end{align}
This choice of $\rA$ and $\rAt$ ensures $(\rA, \rAt)\perp \rXAt$, which enhances the selection power as validated by experiments in Section~\ref{sec.simu.normal}. Moreover, when $\Sigma_X$ is unknown, one can see a good estimation of $\Sigma_X$ leads to a robust FDR control in Section~\ref{sec.fdrcontrol}.

\subsubsection{Construction for Pairwise Comparisons}

\label{sec: construct pairwise comparisons}

In this section, we discuss the case of pairwise comparisons where the distribution of the design matrix $\rX$ can not be approximated by normal distributions. 

For example, 
let the pairwise comparisons be represented by a design matrix $\mX \in \mathbb{R}^{n \times p}$ and responses $\mY \in \mathbb{R}^{n}$. Each row $\rX \in \mathbb{R}^{p}$ of $\mX$ encodes a single comparison between a random pair of items $(j, k)$ where $j<k$, such that its entries are defined as $\rX_{j} = 1$, $\rX_{k} = -1$ and zeros elsewhere. 
Each element $\rY\in \mathbb{R}$ of $\mY$ is modeled from
\begin{equation}
    \Prob(\rY=1|\rX) = f(\rX\beta^*),  \ \Prob(\rY=-1|\rX) = f(-\rX\beta^*),\label{eq: law of y}
\end{equation}
where $\beta^*\in\mathbb{R}^p$ represents the unknown true scores for $p$ objects, and $f(x)$ is a cumulative distribution function, symmetric at 0, i.e. $f(x) = 1 - f(-x)$. To conduct pairwise comparisons, let $D\in\mR^{\frac{p(p-1)}{2}\times p}$ be the graph difference operator of a fully connected graph.\footnote{The construction detailed in this section is applicable when $D$ contains all the edges present in $\mX$. Should this condition not be met, one can simply augment $D$ to include all required edges.} Then pairs with large differences will be selected in model \eqref{eq: gen_int}.

\noindent\textbf{Sequential Constructions.} One direct way to construct the pair $(\rA, \rAt)$ is the sequential construction following Algorithm \ref{algo: sequential generation} with certain choice of $\rA$. 
The detailed construction for the above pairwise comparisons setting is shown in Supplementary Material Section F. However, as will be demonstrated later in Section \ref{sec: pairwise simulations}, these constructions not only incur high computational costs but also exhibit reduced selection power due to the inherent dependency among $(\rXAt, \rA, \rAt)$. This motivates our search for a new approach to construct Split Knockoff copies. Such an approach would not only reduce computational expense but also enhance selection power by ensuring $\rXAt$ is independent of $\rA$, a characteristic previously demonstrated to be advantageous for selection power in the case of normal distributions.

\noindent\textbf{Bootstrap+ Constructions.} Now we will present a novel method of generating Split Knockoff copies, initiated by an augmented bootstrap resampling procedure, so we term it as the bootstrap+ construction. 
First, each row $(\rX^r, \rY^r)$ of the new design matrix $\mX^r$ and response vector $\mY^r$ is a bootstrap sampling from the distribution of $(\rX, \rY)$ and augmented by zeros by:
\begin{enumerate}
    \item $\Prob(\rX^r = 0_p) = \Prob(\rX^r \neq 0_p) = \frac{1}{2}$. \footnote{One can not use a general probability $\Prob(\rX^r = 0_p) = p_0\neq 1/2$ here since $p_0\neq 1/2$ will destroy the exchangeability between $\tilde{A}$ and $A$.}
    \item When $\rX^r = 0_p$, $\rY^r$ is generated through \eqref{eq: law of y}, i.e. $\Prob(\rY^r = 1) = \Prob(\rY^r = -1) = \frac{1}{2}$.
    \item When $\rX^r \neq 0_p$, $(\rX^r, \rY^r)$ takes a random row from $(\mX, \mY)$ respectively. 
\end{enumerate}
This resampling enables the following batch construction of the Split Knockoff copies based on $(\rX^r, \rY^r)$ at a much lower computational cost.

\begin{enumerate}
    \item[1.] The random vector $\rA\in\mathbb{R}^{m}$, where $m = \frac{p(p-1)}{2}$, is generated to satisfy $\rXAt_r :=\rX^r-\rA D = 0_{p}$, i.e. if $X^r=0_p$ then $A=0_m$; if $X^r$ takes a random row $i$ of $\mX$ then $A=e_i$ selects the same row in $D$, where $\{e_i\}$ is the canonical basis of $\mR^{m}$. 
    \item[2.] The random vector $\rAt$ is generated based on $\rA$ by complementarily selecting zero when $\rA$ is non-zero, i.e. $\Prob(\rAt = 0_m|\rA\neq 0_m) = 1$, or uniformly from any canonical basis vector $e_i$ when $\rA$ is zero, i.e. $\Prob(\rAt = e_i|\rA = 0_m) = \Prob(\rA = e_i|\rA \neq 0_m)$.
\end{enumerate}
With the above construction, one can verify that $\rXAt_r \equiv 0_p$, while $[\rA, \rAt]\in\mathbb{R}^{2m}$ is the one-hot vector with one element being 1 and other elements being zeros. Moreover, $\Prob(\rA_j=1) = \Prob(\rAt_j=1)$ for all $j$, which leads to the exchangeability.

This procedure clearly lets $\rXAt_r$ independent from $\rA$, which could help improve the selection power as demonstrated in the Normal distribution case. Simulation experiments in Section \ref{sec: pairwise simulations} validate the selection power improvement with less computational cost compared with the sequential construction. Moreover, the bootstrap+ construction only requires evaluation of $\Prob(\rA = e_i|\rA \neq 0_m)$, which can be directly computed from the \textbf{known} empirical distribution $P_{\mX}$ of the random vector $\rX$. This leads to exact FDR control to be discussed later in Section \ref{sec.fdrcontrol}.

\subsection{Computation of $\rW$-statistics}
\label{sec: w-statistics}
With the above construction, we are now ready to compute the routine $\rW$-statistics
\begin{align}
    \rW_i = w_i\left( (\rXAt, \rA, \rAt), \rY \right),
\end{align}
for $i\in\{1, 2, \cdots, m\}$, similar to \citep{barber2015controlling, candes2018panning}. Similarly, the $\rW$-statistics needs to satisfy the \textit{flip-sign} property:
\begin{align}
\label{eq:swapW}
w_i \left( ( \rXAt, \rA, \rAt)_{\text{swap}(S)}, \rY \right) = 
\begin{cases} 
w_i \left( ( \rXAt, \rA, \rAt), \rY \right), & i \notin S, \\
-w_i \left( ( \rXAt, \rA, \rAt), \rY \right), & i \in S,
\end{cases}
\end{align}
for any $S\subseteq\{1, 2, \cdots, m\}$ and any $i\in\{1, 2, \cdots, m\}$.

There are many approaches, as shown in \citep{barber2015controlling, candes2018panning}, in constructing such $\rW$-statistics satisfying the above property. In this paper, we adopt the following specific construction proposed by \citep{cao2024controlling} that searches in the relaxed neighborhood of the subspace constraint to optimize the selection power. For hyperparameters $\lambda>0$, $\nu>0$, we optimize:
\begin{align}
    \label{minopt:W}
    (\beta(\lambda), \gamma(\lambda), \wt{\gamma}(\lambda)) & =\argmin_{\beta, \gamma, \wt{\gamma}} L(\mXAt\beta + \mA\gamma + \mAt\wt{\gamma}, y)+\cdots \nonumber\\
    &\cdots+ \frac{1}{\nu}(\|D\beta - \gamma\|_2^2 + \|D\beta - \wt{\gamma}\|_2^2)+ \lambda (\|\gamma\|_1 + \|\wt{\gamma}\|_1),
\end{align}
where the loss function $L$ is related to the model or link function, e.g., mean squared loss for linear regression, and cross-entropy loss for logistic regression.
The $\rW$-statistics is constructed as $
    \rW = |\gamma(\lambda)| - |\wt{\gamma}(\lambda)|$. 
    
The hyperparameters $\lambda$ and $\nu$ can be optimized to maximize the selection power. Here, $\lambda$ is the well-known $\ell_1$ regularization parameter. Meanwhile, $\nu$ is the parameter which controls the Euclidean relaxation gap between $D\beta$ and $\gamma$, and is previously introduced in Split LASSO \citep{cao2024controlling} in transformational sparsity problems. The effect of the hyperparameter $\nu$ for improving the selection power will be discussed by simulation experiments in Supplementary Material Section G.

With the above $\rW$-statistics, we are now ready to give the selectors of Model-X Split Knockoff. Let $q>0$ be our target FDR, two data-dependent threshold rules on a pre-set nominal level $q$ are defined as 
\begin{flalign*}
  &\bullet\mbox{(Model-X Split Knockoff)}\ \ \ T_q=\min\left\{\lambda\in\mathcal{W}:\frac{|\{i:\rW_i\le-\lambda\}|}{1\vee|\{i:\rW_i\ge \lambda\}|}\le q\right\},&
  \end{flalign*}
  \begin{flalign*}
  &\bullet\mbox{(Model-X Split Knockoff+)}\ \ \ T_q^+=\min\left\{\lambda\in\mathcal{W}:\frac{1+|\{i:\rW_i\le-\lambda\}|}{1\vee|\{i:\rW_i\ge \lambda\}|}\le q\right\},&
\end{flalign*}
or $T_q,T_q^+ =+\infty$ if the respective set is empty, where $\mathcal{W} = \{|\rW_j|: j = 1, 2, \cdots, m\}\backslash\{0\}$.
The selectors for the Model-X Split Knockoff (+) are  respectively defined as $\wh\cS=\{i:\rW_i\ge T_q\} \mathrm{\ or\ } \wh\cS=\{i:\rW_i\ge T_q^+\}$.

\subsection{Connections with Model-X Knockoff}
The Model-X Split Knockoff is a generalization of the Model-X Knockoff method into the problem of transformational sparsity, whose application includes, but is not limited to pairwise comparisons, wavelet transforms, fused LASSO, and trend filtering. Next, we argue that the Model-X Split Knockoff can achieve better selection power compared with the Model-X Knockoff when both are applicable.

To be specific, consider the case of $D = I_p$, the model \eqref{eq: gen_int} degenerates to the traditional variable selection problem that Model-X Knockoff can deal with. In this case, we show by the following proposition that our method achieves at least the same selection power as the Model-X Knockoff.

\begin{proposition}
\label{prop: higher selection power}
    Consider the case $D = I_p$ in the model \eqref{eq: gen_int}, there exists a pair of Model-X Split Knockoff copies $(\rA, \rAt)$, such that Model-X Split Knockoff with $(\rA, \rAt)$ achieves at least the same selection power as the Model-X Knockoff.
\end{proposition}
The proof of this proposition is straightforward. Take $\rA = \rX$ in the construction of the Model-X Split Knockoff, then $\rXAt = \rX-\rA I_p = 0$. Under this condition, the model \eqref{eq:decomp} and the exchangeability condition \eqref{eq: exchangeability} degenerate precisely into the Model-X Knockoff framework. Therefore, Model-X Split Knockoff with this choice of $\rA$ have the same selection power with the Model-X Knockoff. 

As shown in Section \ref{sec: construction}, taking $\rA = \rX$ might not be the optimal choice in the Model-X Split Knockoff in terms of the selection power, that our specific choice in Section \ref{sec: construction} may achieve better selection power. Such a claim is further validated by simulation experiments in Section~\ref{sec.simu.normal} and the Alzheimer's disease study in Section~\ref{sec.alz}.

\section{FDR Control of Model-X Split Knockoff}
\label{sec.fdrcontrol}

In this section, we provide the robust FDR control guarantees under the proposed Model-X Split Knockoff framework. To be specific, we apply the ``leave-one-out'' technique introduced in \cite{barber2020robust} to establish robust FDR control when the distribution of $X$ can be estimated with a bounded deviation on the KL-divergence.

In practice, there are many scenarios where the distribution of $\rX$ is unknown and needs to be estimated. 
To be specific, let $P_X$ be the estimated distribution\footnote{Here, $P_X$ and $P^*_{X}$ represent the probability mass function for discrete random variables and probability density function for continuous random variables.} of $\rX$, and $P^*_{X}$ be the true distribution of $\rX$. Consider the following  sample Kullback–Leibler (KL) divergence for $j\in\{1, 2, \cdots, m\}$:
\begin{equation}
\label{eq: KL}
\widehat{\mathrm{KL}}_j: = \sum_{i=1}^n\log{\left[ \frac{P^*_{X}(\mXAt_i + \mA_i D)}{P_{X}(\mXAt_i + \mA_i D)}\frac{P_{X}\left\{\mXAt_i+ \mA_i(j)D \right\}}{P^*_{X}\left\{\mXAt_i + \mA_i(j) D \right\}} \right]},
\end{equation}
where $\mXAt_i$, $\mA_i$, $\mAt_i$ represent the $i$-th row of $\mXAt$, $\mA$, $\mAt$ receptively, and $\mA_i(j)$ represents the vector obtained by substituting the $j$-th element of $\mA_i$ with the $j$-th element of $\mAt_i$. Clearly, if $P_X = P^*_{X}$, e.g. the distribution of $X$ is known, the KL divergence $\widehat{\mathrm{KL}}_j = 0$ for $j\in\{1, 2, \cdots, m\}$. Meanwhile, in the cases where the true distribution $P^*_X$ is unknown and is estimated by $P_X$,  $\widehat{\mathrm{KL}}_j$ measures the impact of the estimation error to the $j$-th feature based on the observations $\mXAt, \mA, \mAt$. 

Now we are ready to state the main theorem on the robust false discovery control of the Model-X Split Knockoffs.

\begin{theorem}
\label{thm.robust}
Suppose the set of random vectors $(\rXAt, \rA, \rAt)$ satisfy the exchangeability \eqref{eq: exchangeability} when $X\sim P_X$. Then when $X\sim P^*_X$, for any $q>0$, $\epsilon \ge 0$, the following holds.
\begin{enumerate}
    \item For the Model-X Split Knockoff, there holds $
    \mE \left( \frac{|\{j: j\in \wh \cS \cap \mathcal{H}_0^{\gamma} \text{ and } \widehat{\mathrm{KL}}_j \leq \epsilon \}|}{|\wh \cS| + q^{-1}} \right) \leq q \cdot e^\epsilon$, 
leading to the following bound of $\mathrm{mFDR}$ defined in \cite{barber2015controlling}, 
\begin{equation*}
    \mathrm{mFDR} = \mE \left( \frac{|\{j: j\in \wh \cS\cap \mathcal{H}_0^{\gamma}  \}|}{|\wh \cS| + q^{-1}} \right) \leq \min_{\epsilon \geq 0} \left\{ q \cdot e^\epsilon + \Prob\left(\max_{j \in \mathcal{H}_0^{\gamma}} \widehat{\mathrm{KL}}_j > \epsilon\right) \right\}.
\end{equation*}
\item For the Model-X Split Knockoff+, there holds $
    \mE \left( \frac{|\{j: j\in \wh \cS \cap \mathcal{H}_0^{\gamma} \text{ and } \widehat{\mathrm{KL}}_j \leq \epsilon \}|}{|\wh \cS| \vee 1} \right) \leq q \cdot e^\epsilon$, 
which leads to the following bound of $\mathrm{FDR}$,
\begin{equation*}
    \mathrm{FDR} = \mE \left( \frac{|\{j: j\in \wh \cS \cap \mathcal{H}_0^{\gamma} \}|}{|\wh \cS| \vee 1} \right) \leq \min_{\epsilon \geq 0} \left\{ q \cdot e^\epsilon + \Prob\left(\max_{j \in \mathcal{H}_0^{\gamma}} \widehat{\mathrm{KL}}_j > \epsilon\right) \right\}.
\end{equation*}
\end{enumerate}
The proof of Theorem \ref{thm.robust} will be presented in Supplementary Material Section B.
\end{theorem}

Theorem \ref{thm.robust} states that the loss in the FDR control vanishes when the estimation on the distribution of $X$ becomes accurate. For the case that $P_X = P^*_X$, taking $\epsilon = 0$ in Theorem \ref{thm.robust} leads to the exact FDR control. In the following, we will give one concrete example of such exact FDR control in pairwise comparisons, and one concrete example of robust FDR control when the distribution of $X$ is normal.

\noindent{\textbf{Exact FDR Control Example: Pairwise Comparisons.}} While the condition $P_X = P^*_X$ is generally unachievable in practice, there are some situations where such a condition holds. For instance, consider the canonical example of pairwise comparisons in Section \ref{sec: construct pairwise comparisons}. In that case, the new design $X^r$ constructed by the \emph{bootstrap+} construction is based on boostrapping (augmented with zeros) the empirical distribution $P_{\mX}$ of the random vector $\rX$. Therefore, it is straightforward to verify that $(\rXAt_r, \rA, \rAt)$ satisfies the exchangeability \eqref{eq: exchangeability} as $P_{X^r} = P_{X^r}^* = g(P_{\mX})$ for some suitable function $g$, and the exact false discovery rate control is achieved.

\noindent{\textbf{Robust FDR Control Example: Normal Distributions.}} In the following, we provide concrete computation of $\max_j \wh{\mathrm{KL}}_j$ when $\rX$ follows the normal distribution $\mathcal{N}(\mathbf{0}_p,\Sigma_X)$. In this example, we construct the Model-X Split Knockoff as in \eqref{eq.A-tilde-general}.

\begin{theorem}
    \label{thm.KL-gauss}
    Denote $\Theta:=\Sigma^{-1}_X$ and $\wh{\Theta}:=\wh{\Sigma}^{-1}_X$, where $\wh{\Sigma}_X$ is an estimation of $\Sigma_X$. Suppose $\rX \sim \mathcal{N}(\mathbf{0}_p,\Sigma_X)$. Let $\Lambda_X$ the maximum eigenvalue of $\Sigma_X$, and define
    \begin{align*}
        \delta_\Theta & := 4\alpha \Vert D_j \Vert_2 \Vert \Delta \Vert_2 \left( \sqrt{2p \Lambda_X} + 2\alpha \Vert D_j \Vert_2^2 \right)
    \end{align*}
    for any choice of $\alpha > 0$ such that $\wh{\Sigma}_X - \alpha D^\top  D \succeq 0$ and $A \sim \cN(0,\alpha I_m)$.
    If $\frac{\log{(m+p)}}{n} = o(1)$, with probability at least $1 - \frac{2}{p} - \frac{2}{m}$, there holds 
    \begin{align*}
        \max_j \wh{\mathrm{KL}}_j \leq 2\delta_\Theta \sqrt{n\log{m}} \{1 + o_p(1)\}. 
    \end{align*}
\end{theorem}

\begin{remark}
    In many scenarios, $\Vert D_j \Vert_2$ is at constant levels. Therefore, the estimation error $\Vert \Theta - \wh{\Theta} \Vert_2$ being $o_p\left(\frac{1}{\sqrt{n\log p}}\right)$ is sufficient to make $\max_j \wh{\mathrm{KL}}_j$ sufficiently small. Unlabeled data can help to achieve this goal with details being discussed in the remark of Supplementary Material Section E.
\end{remark}

\section{Simulation Experiments}
\label{sec.simu}
In this section, we evaluate the Model-X Split Knockoff method in two canonical settings, the normal distribution setting and the pairwise comparisons setting. The simulation results show that the Split Knockoff method achieves the desired false discovery rate control as well as relatively high selection power, specifically higher compared with Model-X Knockoff when both are applicable.

\subsection{Normal Distribuions}
\label{sec.simu.normal}
\noindent \textbf{Experimental Setting:}
In this section, simulation experiments are conducted under the case where the distribution of the design $\rX$ follows the normal distribution. In particular, rows of the design matrix $\mX  \in \mathbb{R}^{n \times p}$ is generated independent and identically distributed (i.i.d.) from a multivariate normal distribution $\mathcal{N}(0_p, {\bSigma})$, where the covariance matrix ${\bSigma}$ has the correlation structure ${\bSigma}_{i,j} = c^{|i-j|}$ for all $(i,j)$ pairs with correlation parameter $c = 0.5$. 
The true coefficient vector ${\beta}^* \in \mathbb{R}^p$ is defined by: $\beta^*_i = A$ if $i \le k, i \equiv 0, -1\pmod 3$; and $\beta^*_i=0$ otherwise. 
In this section, we consider the logistic regression, where the response $\mY\in \mathbb{R}^n$ is generated from $\mY \sim \text{Bernoulli}\left\{ \sigma(\mX\beta^*) \right\}$, where $\sigma (x):= \frac{1}{1+\exp(-x)}$.

The transformational sparsity is specified by the linear transformation matrix $D \in \mathbb{R}^{m \times p}$ such that $\gamma^{*} = D \beta^{*}$ is sparse. Based on the above specific choice of $\beta^{*}$, the following three kinds of $D$ are considered in this simulation:

\begin{itemize}
    \item $\beta^{*}$ itself is sparse, so that one can take $D_1 = I_p$, leading to $m = p$.
    
    \item $\beta^{*}$ is a uni-dimensional piecewise constant function, so that one can take $D_2 \in \mathbb{R}^{(p-1) \times p}$ as: $ D_2(i, i) = 1 $, $ D_2(i, i+1) = -1 $ for $ i \in \{1, \dots, p-1\} $, and $ D_2(i, j) = 0 $ for other pairs of $(i, j)$. In this case, $ m = p - 1 < p $.
    
    \item Combining the above two cases, one can take $ D_3 = \left( D^\top_1, D^\top_2 \right)^\top  \in \mathbb{R}^{(2p-1) \times p} $, where $ m = 2p - 1 > p $.
\end{itemize}
Throughout this section, the target FDR level $q$ is set to be $0.2$, while the coefficients in the above simulation design are taken to be $p = 100$, $k=20$, and $A = 1$. For the case where $D_1 = I_p$, the performance of Model-X Knockoffs\footnote{The Model-X Knockoff method is implemented using the \href{https://github.com/msesia/knockoff-filter}{official repository}.} is also evaluated for comparisons. The Model-X (Split) Knockoff copies are generated from the data matrix $\X$ through the estimated covariance matrix $\wh{\bSigma} = \X^\top \X / n$. Specifically, the Split Knockoff copies $\rA$ and $\rAt$ are generated through Equation \eqref{eq.A-tilde-general}. 

\begin{figure}[ht]
    \centering
    \subfigure[$D = D_1$ (MK) ]{
    \begin{minipage}[t]{0.225\textwidth}
    \centering
    \includegraphics[width=1\textwidth]{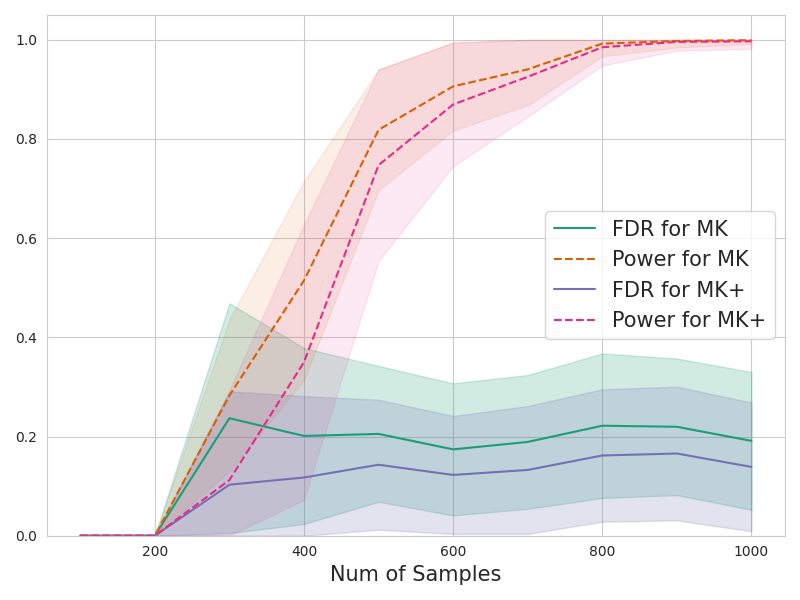}
    \end{minipage}%
    }
    \subfigure[$D = D_1$ (MSK)]{
    \begin{minipage}[t]{0.225\textwidth}
    \centering
    \includegraphics[width=1\textwidth]{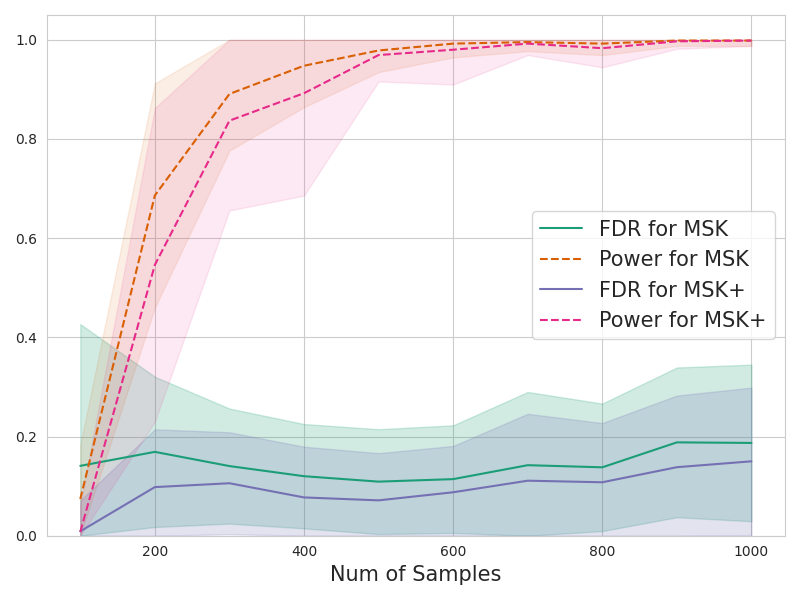}
    \end{minipage}%
    }
    \subfigure[$D = D_2$ (MSK)]{
    \begin{minipage}[t]{0.225\textwidth}
    \centering
    \includegraphics[width=1\textwidth]{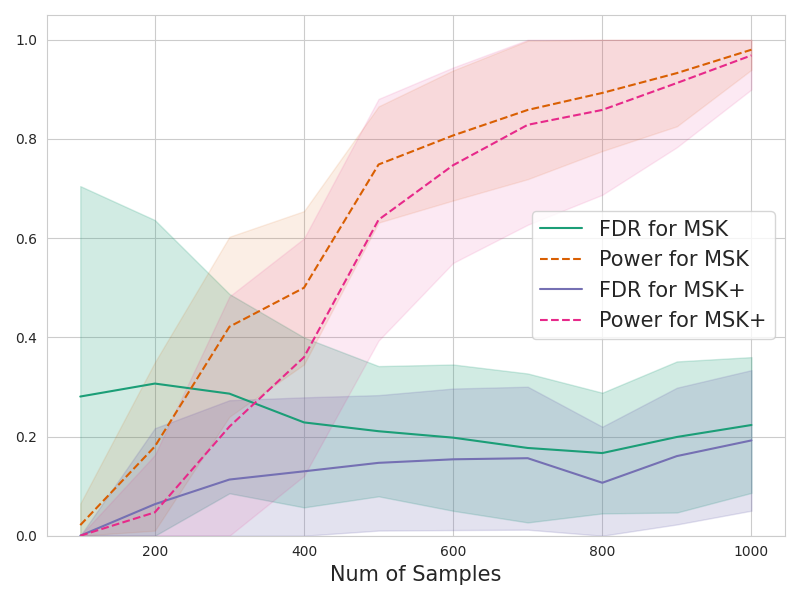}
    \end{minipage}%
    }
    \subfigure[$D = D_3$ (MSK)]{
    \begin{minipage}[t]{0.225\textwidth}
    \centering
    \includegraphics[width=1\textwidth]{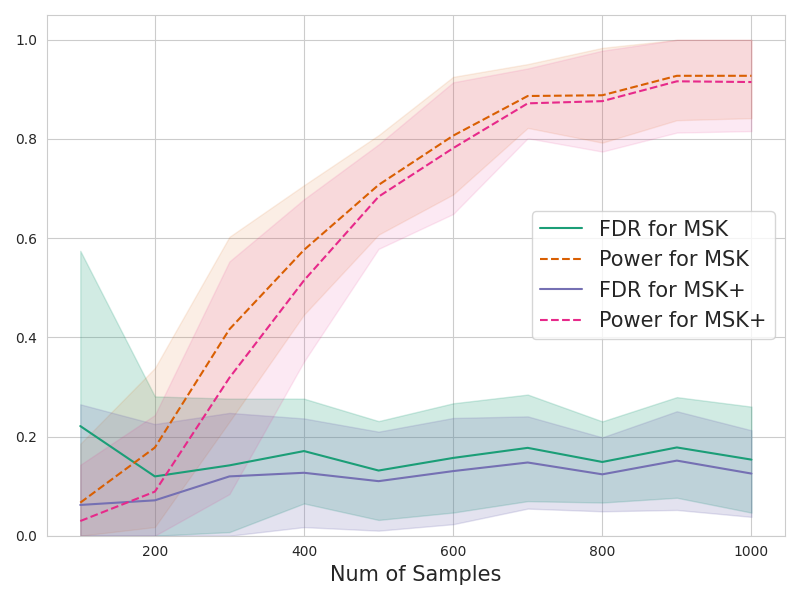}
    \end{minipage}%
    }
    \caption{Performance of Model-X Knockoffs (when applicable) and Model-X Split Knockoffs when $q = 0.2$ and the number of samples varies. The curves indicate the mean performance, with the shaded areas representing the standard deviation across 200 simulations. For brevity, we denote Model-X Knockoffs and Model-X Split Knockoffs as MK and MSK, respectively.}
    \label{fig:numplot comparisons logistic}
\end{figure}

In the first two figures of Figure \ref{fig:numplot comparisons logistic}, we compare the performance of Model-X Knockoffs and Model-X Split Knockoffs when $D = D_1 = I_p$. In both figures, with the increase of sample size, the selection power of both methods gradually increases from 0 to 1, while the FDR of both methods are under control for all sample sizes. In terms of selection power, the Model-X Split Knockoff demonstrably yields superior results across all tested sample sizes, a finding supported by Proposition \ref{prop: higher selection power}. We hypothesize that this performance gain is a direct consequence of the improved incoherence conditions, which arise from the orthogonal design of $\rA$ and $\rAt$ in Equations \eqref{eq.A-tilde-general}.

In the last two figures of Figure \ref{fig:numplot comparisons logistic}, the performance of Model-X Split Knockoffs is evaluated in the scenarios where $D$ is non-trivial, i.e., $D\neq I_p$. In both cases, the FDR of Model-X Split Knockoffs (MSK+) is under control universally for all choices of sample sizes, with MSK higher as defined. Meanwhile, the selection power of Model-X Split Knockoffs follows a similar trend as the case where $D=I_p$, i.e., the selection power gradually increases from 0 to 1 when the sample size enlarges.

\subsection{Pairwise Comparisons}
\label{sec: pairwise simulations}

\noindent \textbf{Experimental Settings:} In this section, the pairwise comparison problem among $p$ objects with $n$ noisy observations of comparisons is considered. The scores for $p$ objects, denoted as $\beta^* \in \mathbb{R}^p$, are generated from a standard normal distribution $\mathcal{N}(0, 1)$. To introduce sparsity into the dataset, the scores in $\beta^*$ are truncated according to a sparsity ratio $k$ set at 0.5. Specifically, only the first $k \cdot p$ components of $\beta^*$ are retained, while the remaining components are set to zero.

The rows of the pairwise comparison design matrix $\X \in \{0, 1, -1\}^{n \times p}$ are generated independently. For the $i$-th row, two indices $1\le j<k\le p$ are chosen uniformly at random with probability $\frac{2}{p(p-1)}$ for each possible pair of $(j, k)$. Then we set $\mX_{i, j} = 1$, $\mX_{i, k} = -1$, and $\mX_{i, l} = 0$ for $l\neq j, k$. This specific construction ensures that each row of $\mX$ represents a (noisy) comparison between two objects. 

In this simulation experiment, the performance of Model-X Split Knockoff is evaluated under the Bradley-Terry model, where the response vector $\mY \in \mathbb{R}^n$ is generated by the logistic regression $\mY \sim \text{Bernoulli}\left\{ \sigma(\mX\beta^*) \right\}$. The transformational matrix $D \in \mathbb{R}^{\frac{p(p-1)}{2} \times p}$ is taken to be the graph difference operator on a fully connected graph with $p$ vertices to conduct pairwise comparisons.

\begin{figure}[ht]
    \centering
    \subfigure[Bootstrap+]{
    \begin{minipage}[t]{0.3\textwidth}
    \centering
    \includegraphics[width=0.8\textwidth]{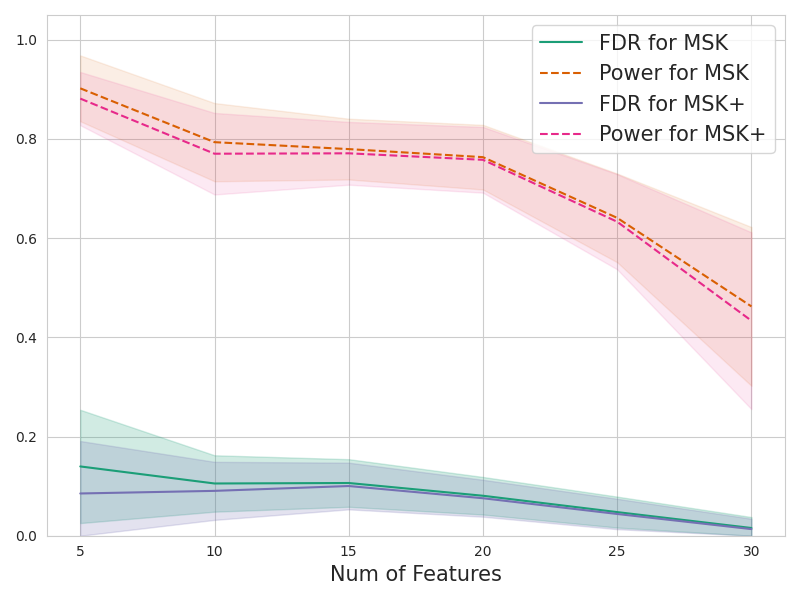}
    \end{minipage}%
    }
    \subfigure[Sequential]{
    \begin{minipage}[t]{0.3\textwidth}
    \centering
    \includegraphics[width=0.8\textwidth]{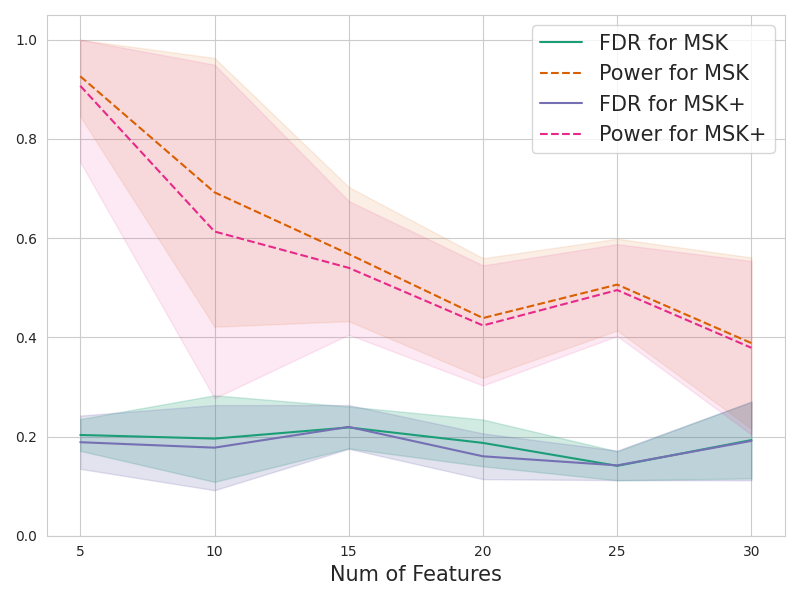}
    \end{minipage}%
    }
    \subfigure[Runtime (Second)]{
    \begin{minipage}[t]{0.3\textwidth}
    \centering
    \includegraphics[width=0.8\textwidth]{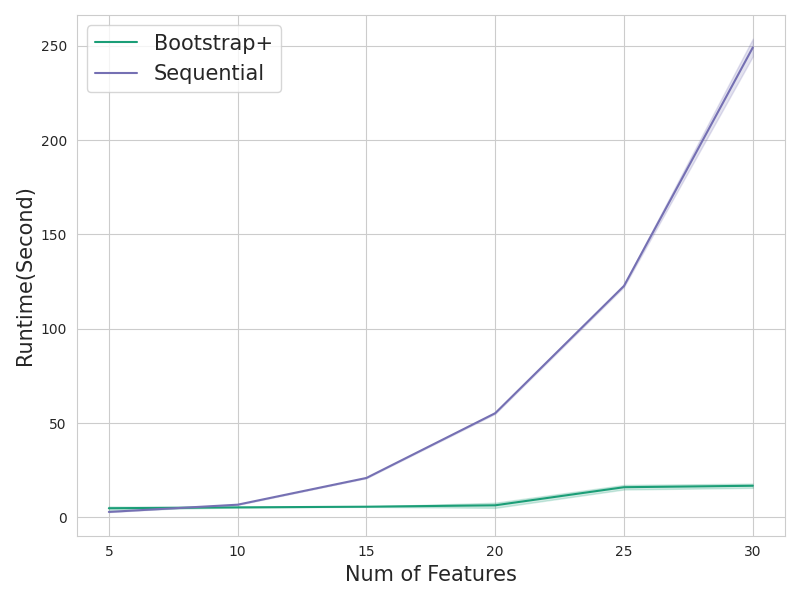}
    \end{minipage}%
    }
    \caption{Performance and runtime of Model-X Split Knockoff when $q = 0.2$ as the number of features varies for bootstrap+ and sequential constructions. The curves indicate the mean performance, with the shaded areas representing the standard deviation across 200 simulations.}
    \label{fig:featplot pairwise comparisons}
\end{figure}

Figure \ref{fig:featplot pairwise comparisons} compares the performance of bootstrap+ constructions of Model-X Split Knockoff copies against the sequential constructions. It shows that, although the two types of constructions achieve desired FDR control, the selection power of the bootstrap+ constructions is higher compared with the sequential constructions for all choices of the number of features. This is explained in Section \ref{sec: construct pairwise comparisons} that the bootstrap+ constructions makes sure $\rXAt$ to be independent from $\rA$, which helps improve the incoherence conditions and thus the selection power.

Moreover, when $p$ is large, the sequential construction may need to estimate too many conditional distributions, so the computation cost explodes. On the other hand, the bootstrap+ constructions do not suffer from this issue and handle such a case easily. As shown in Figure~\ref{fig:featplot pairwise comparisons}, the sequential method exhibits a much steeper increase in running time compared to the bootstrap+ method, which requires significantly less computation time under the same settings. 

Due to the above two points, the bootstrap+ construction will be the default construction to implement Model-X Split Knockoff for pairwise comparisons.

\section{Alzheimer’s Disease}
\label{sec.alz}
We employ the Model-X Split Knockoff method in this study to analyze lesion regions and their connectivity within the brains of individuals affected by Alzheimer's Disease (AD), which has been the subject of growing research focus.

We apply our method to the structural Magnetic Resonance Imaging (MRI) in the \href{http://adni.loni.ucla.edu}{ADNI} dataset. We extracted 752 samples that are composed of 126 AD, 433 Mild Cognitive Impairment (MCI), and 193 Normal Controls (NC). Applying Automatic Anatomical Labeling (AAL), we extract the volumes of $90$ processed brain regions for all $752$ samples. Therefore, for our design matrix $\mX \in \mathbb{R}^{n \times p}$, \(\mX_{i,j}\) denotes the column-wise normalized volume of region \(j\) for sample \(i\). Besides, \(\mY \in \mathbb{R}^n\) is taken as the Alzheimer's Disease Assessment Scale (ADAS), a measure initially developed to evaluate the severity of cognitive dysfunction~\citep{rosen1984new}, which was later shown to effectively distinguish between clinically diagnosed Alzheimer's disease and normal controls~\citep{zec1992alzheimer}.

\subsection{Region Selection}

The region selection results of Model-X Split Knockoff and Model-X Knockoff are shown in Table~\ref{tab:selected_regions}. Overall, Model-X Split Knockoff identifies 12 brain regions associated with Alzheimer's disease, whereas the standard Model-X Knockoff selects only 8, demonstrating the enhanced selection power of our approach.

\begin{table}[ht]
    \centering
    \caption{Selected Regions by Model-X Knockoff and Model-X Split Knockoff on Alzheimer’s Disease ($q = 0.2$). The full names of the regions are provided in Supplementary Material Section I.}
    \resizebox{\textwidth}{!}{
    \begin{tabular}{ccccccccccccccc}
    \toprule
    \textbf{Region} & IPL R & ORBsup L & MFG R & ROL L & SFGmed R & PCG L  & PCG R & HIP L & HIP R & LING L&FFG R & MTG L   & MTG R & ITG L\\
    \midrule
    \textbf{Model-X Knockoff} & \checkmark &  & \checkmark &  &  &  &           & \checkmark & \checkmark & \checkmark       & \checkmark   & \checkmark & \checkmark &  \\
    \textbf{Model-X Split Knockoff} & & \checkmark & \checkmark & \checkmark & \checkmark & \checkmark   & \checkmark  & \checkmark & \checkmark & \checkmark      & \checkmark  & \checkmark & & \checkmark \\
    \bottomrule
    \end{tabular}
    }
    \label{tab:selected_regions}
\end{table}

Both methods select the hippocampus (``HIP L", ``HIP R"), middle temporal lobes (``MTG L", ``MTG R"), middle frontal gyrus (``MFG R"), and several cerebral cortex areas such as the lingual gyrus (``LING L", ``LING G") and fusiform gyrus (``FFG R"). These regions have been previously reported as potential biomarkers for Alzheimer's disease in the literature \citep{lee2020posterior, peters2009neural, yang2019study}. Most of additional regions identified by our method have been reported to show early degeneration in previous studies \citep{van2000orbitofrontal,cajanus2019association,pengas2010focal,lee2020posterior}, with the exception of the Superior frontal gyrus, orbital part (“ORBsup L”), which may represent a false discovery. In particular, significant structural changes in nodal centrality within the rolandic operculum (``ROL L") have been observed in the AD population \citep{yao2010abnormal,bajo2015scopolamine}. Besides, the posterior cingulate cortex (``PCG L", ``PCG R"), which is involved in episodic memory, was found to exhibit atrophy at the earliest clinical stages of sporadic AD \citep{lee2020posterior}.

\subsection{Connection Selection}
In this section, we detect abnormal connections between brain regions. We represent brain connectivity as a graph $G = (V, E)$, where $V$ denotes the brain regions and $E$ represents edges connecting adjacent regions. Based on this formulation, the transformational matrix $D$ is defined as a graph difference operator, such that $(D\beta)_{(i,j)} = \beta_i - \beta_j$ for adjacent brain regions $(i, j)$. Since adjacent brain regions typically show similar activity, a significant signal on an edge suggests connections between relatively stable regions and severely atrophied regions contributing to the disease. The abnormal connections identified by our method are presented in Table~\ref{tab:selected_connect}.

\begin{table}[ht]
    \centering
    \caption{Selected Connections by Model-X Split Knockoff on Alzheimer’s Disease ($q = 0.2$). The full names of the regions are provided in Supplementary Material Section I.}
    \resizebox{.7\textwidth}{!}{
    \begin{tabular}{c|c|c|c}
       \toprule
        HIP L \& PHG L&
        HIP L \& LING L&
        HIP L \& FFG L&
        HIP L \& PUT L \\
        HIP L \& HES L&
        HIP R \& PHG R&
        HIP R \& TPOmid R&
        HIP R \& ITG R \\
        HIP L \& INS L&
        LING L \& FFG L&
        LING R \& PHG R&
        ROL L \& HES L  \\
        MTG L \& INS L&
        MTG L \& MOG L&
        MTG R \& TPOsup R&
        PCG R \& PCUN R \\
        PCL R \& IPL R & & & \\
        \bottomrule
    \end{tabular}}
    \label{tab:selected_connect}
\end{table}

Notably, many of the detected connections involve the hippocampus (``HIP L", ``HIP R"), which exhibits the earliest and most significant atrophy in Alzheimer's disease (AD) \citep{greicius2003functional,yang2019study,cajanus2019association,yao2010abnormal}. 

In addition, some connections involve the middle temporal lobes (``MTG L", ``MTG R"). The middle temporal lobe, a key component of the temporal cortex, plays an essential role in semantic memory and language processing, and is implicated in the early pathological changes of Alzheimer’s disease (AD) \citep{chen2022spatially}. Moreover, the lingual regions (``LING L", ``LING R") were involved in the selected connections. This finding is consistent with previous studies showing that the lingual gyrus—associated with visual processing and visual memory—is important in the early detection of Alzheimer's disease.

\section{Pairwise Comparison: WorldCollege Dataset}
\label{sec.college}
In this section, we apply our method to \href{http://www.allourideas.org/worldcollege}{the World College dataset}. For our experiments, we subsample the dataset to exclude colleges with very few annotated comparisons. Our specific sub-dataset comprises $p=69$ colleges that rank in the top 80 of the QS rankings among the original 261 colleges. This subsampling yields $n=2,112$ annotated pairwise comparisons among these 69 colleges, resulting in $m=608$ unique pairwise comparisons. More details on the descriptions of dataset and the QS rankings are in Supplementary Material Section I.

In this experiment, the Bradley-Terry model is adopted, and thus $D\in\mathbb{R}^{m\times p}$ is defined as: for each pair of college ($i_1<i_2$) compared by the annotator, there exists an unique row $j$ in $D$, such that $D_{j, i_1} = 1$, $D_{j, i_2} = -1$, and $D_{j, k} = 0$ for $k\neq i_1, i_2$.

\begin{figure}[ht]
    \centering
    \includegraphics[width=\textwidth]{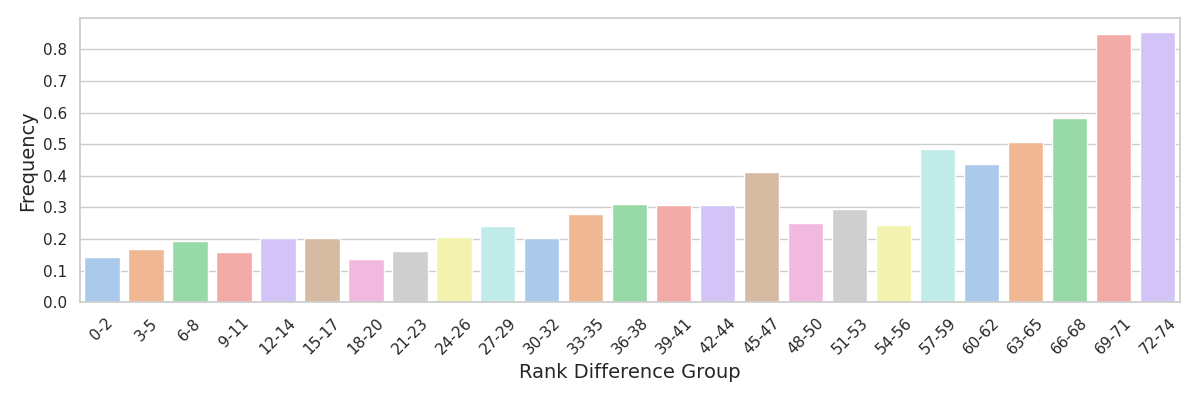}
    \caption{Selection Frequency of Pairwise Comparisons by Model-X Split Knockoffs Across 100 Runs. The bar chart displays the average selection frequency for college pairs identified by our method. Each bar corresponds to pairs within $D$ whose QS ranking differences are categorized by the range shown on the $X$-axis.}
    \label{fig:col_prob_sampled}
\end{figure}

Figure \ref{fig:col_prob_sampled} presents the selection frequency from 100 runs of our method, showing a clear positive correlation with the QS ranking difference between pairs. This trend demonstrates that our method effectively identifies university pairs with significant ranking gaps, which suggests a low false discovery rate. For instance, connections between colleges with large rank differences, such as Massachusetts Institute of Technology (1) and Osaka University (63), or Harvard University (rank 3) and Tohoku University (rank 75), were detected in 73 and 89 out of 100 experiments respectively. The consistent selection of pairs with rank differences over 60 underscores the method's sensitivity in uncovering substantial disparities.

\section{Conclusion}

We introduced the Model-X Split Knockoff method, a framework extending both Split and Model-X Knockoffs to control the FDR for transformation selection in a broad range of statistical models. Our key innovation, an auxiliary randomized design, resolves the challenge of combining random designs with deterministic transformations. This approach not only enables inference on transformational sparsity but also provides superior selection power over standard Model-X Knockoffs. The method's effectiveness is validated by simulations and real-world applications.

\bibliography{ref}

\end{document}